\documentclass[letterpaper,10pt, conference]{_aux/ieeeconf}
\usepackage{times}
\usepackage{setspace}
\usepackage{url}
\spacing{1}
\usepackage[utf8]{inputenc}
\usepackage[T1]{fontenc}
\usepackage{graphicx}
\usepackage{wrapfig}
\usepackage{sidecap} 
\usepackage[export]{adjustbox}
\usepackage{subcaption}
\usepackage[format=plain,font=small,labelfont=bf,labelsep=period]{caption}
\usepackage{float}
 
\usepackage{hyperref}

\usepackage{amsmath}
\usepackage{amssymb}
\usepackage{amsthm}
\usepackage{mathtools}
\usepackage[normalem]{ulem}
\usepackage{paralist}
\usepackage[space]{grffile}
\usepackage{color}
\let\labelindent\relax
\usepackage{enumitem}
\usepackage{bm}
\usepackage{cancel}
\usepackage{hhline}
\usepackage[c2 , nocomma]{optidef}
\usepackage{oubraces}
\usepackage{algorithmic}
\usepackage{graphicx}
\usepackage{textcomp}

\usepackage{amsfonts}
\usepackage{lipsum}
\usepackage{cleveref} 
\usepackage{cite}

\usepackage{comment}

\newtheorem{remark}{Remark}
\newtheorem{assumption}{Assumption}

\newtheorem{theorem}{Theorem}
\newtheorem{corollary}{Corollary}
\newtheorem{lemma}{Lemma}
\theoremstyle{definition}
\newtheorem{definition}{Definition}
\theoremstyle{remark}
\theoremstyle{definition}
\theoremstyle{definition}

\newtheorem{casestudy}{Case Study}

\newcommand{\R}{\mathbb{R}}

\newcommand{\xaux}{\hat{\boldsymbol{x}}}

\newcommand{\x}{\boldsymbol{x}}

\newcommand{\sspace}{\mkern9mu}
\newcommand{\nspace}[1]{\mkern#1mu}

\newcommand{\Cs}{\mathcal{C}_{\rm S}} 
\newcommand{\Cb}{\mathcal{C}_{\rm B}} 
\newcommand{\Csaux}{\widehat{\mathcal{C}}_{\rm S}}
\newcommand{\Cbaux}{\widehat{\mathcal{C}}_{\rm B}} 
\newcommand{\Cbi}{\mathcal{C}_{\rm I}}
\newcommand{\Cbistar}{\mathcal{C}^\star_{\rm I}}

\newcommand{\Btheta}{\boldsymbol{\theta}}

\newcommand{\Cbiaux}{\widehat{\mathcal{C}}^\star_{\rm I}}

\newcommand{\Tt}{T} 

\newcommand{\ub}{\boldsymbol{k}_{\rm b}}
\newcommand{\kstar}{\boldsymbol{k}_{\rm e}}

\newcommand{\phinb}[2]{\boldsymbol{\phi}_{\rm b} (#1, #2)}
\newcommand{\phinom}{\phinb{\tau}{\boldsymbol{x}}}
\newcommand{\phinomT}{\phinb{\Tt}{\boldsymbol{x}}}

\newcommand{\phins}[2]{\boldsymbol{{\phi}}_{\rm s} (#1, #2)}

\newcommand{\phinbaux}[2]{\boldsymbol{\widehat{\phi}}_{\rm s} (#1, #2)}
\newcommand{\phiaux}{\phinbaux{\tau}{\xaux}}
\newcommand{\phiauxT}{\phinbaux{\Tt}{\xaux}}

\definecolor{darkblue}{RGB}{0,0,102}
\definecolor{lightblue}{RGB}{77,77,148}

\definecolor{gold}{RGB}{234, 170, 0}
\definecolor{metallic_gold}{RGB}{139, 111, 78}

\usepackage{xfrac}

\usepackage{dblfloatfix}

\usepackage{pdfpages}
\usepackage{graphicx}

\usepackage{enumitem}

\newcommand{\argmax}{\operatornamewithlimits{arg\,max}}

\newcommand{\ed}[1]{{\color{blue} ED: $\#$1}}
\def\BibTeX{{\rm B\kern-.05em{\sc i\kern-.025em b}\kern-.08em
T\kern-.1667em\lower.7ex\hbox{E}\kern-.125emX}}

\IEEEoverridecommandlockouts

\begin{document}

\title{\LARGE
\textbf{Generalizations of Backup Control Barrier Functions: \\ Expansion and Adaptation for Input-Bounded Safety-Critical Control}
}

\author{David E. J. van Wijk$^{1}$, Dohyun Lee$^{1}$,
Ersin Da\c{s}$^{2}$, Tamas G. Molnar$^{3}$, Aaron D. Ames$^{1}$, and Joel W. Burdick$^{1}$
\thanks{*This work was supported in part by DARPA under the LINC program and by the Technology Innovation Institute (TII).}
\thanks{$^{1}$Mechanical and Civil Engineering, California Institute of Technology, Pasadena, CA 91125, USA, \texttt{\{vanwijk, dohyun, ames, jburdick\}@caltech.edu}.}
\thanks{$^{2}$Mechanical, Materials, and Aerospace Engineering, Illinois Institute of Technology, Chicago, IL 60616, USA, \texttt{edas2@illinoistech.edu}.}
\thanks{$^{3}$Mechanical Engineering, Wichita State University, Wichita, KS 67260, USA, \texttt{tamas.molnar@wichita.edu}.}
}

\captionsetup{font=footnotesize}

\maketitle

\begin{abstract}
Guaranteeing the safety of nonlinear systems with bounded inputs remains a key challenge in safe autonomy. Backup control barrier functions (bCBFs) provide a powerful mechanism for constructing controlled invariant sets by propagating trajectories under a pre-verified backup controller to a forward invariant backup set. While effective, the standard bCBF method utilizes the same backup controller for both set expansion and safety certification, which can restrict the expanded safe set and lead to conservative dynamic behavior. In this study, we generalize the bCBF framework by separating the set-expanding controller from the verified backup controller, thereby enabling a broader class of expansion strategies while preserving formal safety guarantees. We establish sufficient conditions for forward invariance of the resulting implicit safe set and show how the generalized construction recovers existing bCBF methods as special cases. Moreover, we extend the proposed framework to parameterized controller families, enabling online adaptation of the expansion controller while maintaining safety guarantees in the presence of input bounds.
\end{abstract}
\begin{spacing}{0.95}
\section{Introduction}

Ensuring the safety of modern autonomous systems is crucial to their real-world deployment. \textit{Control barrier functions} (CBFs) have emerged as a popular approach to guaranteeing the safety of such systems by enforcing the forward invariance of safe sets \cite{ames_2017}. While CBFs have been successfully deployed across many applications, from robotics to aerospace systems \cite{long2026sensor,dickson2025safe,agrawal2025online},
many important challenges remain. 
Importantly, one must verify that for all states in the safe set, there exists a \textit{feasible} safe controller, which is especially difficult for complex systems with input bounds.

Recent works have proposed reachability and prediction-based methods for constructing safety certificates that reduce conservatism under input constraints, thereby certifying a larger region of safe operation. Specifically, \cite{choi2021robust} introduces robust control barrier-value functions to unify Hamilton--Jacobi reachability and CBFs, yielding certificates whose safe set recovers the viability kernel while remaining amenable to online CBF-based filtering. Similarly, \cite{tonkens2022refining} uses dynamic-programming-based reachability updates to refine a candidate CBF, producing certificates that become progressively less conservative and converge to a valid barrier function. Along a related predictive direction, \cite{wiltz2023construction} shows that a time-invariant CBF can be constructed from finite-horizon predictions using only a known subset of a controlled invariant set.

Perhaps one of the most promising approaches to address input constraints is \textit{backup control barrier functions} (bCBFs) \cite{gurriet_online_2018,gurriet_scalable_2020}.
Using the idea of set expansion, bCBFs create
controlled invariant safe sets by examining the evolution of the system under a pre-verified \textit{backup controller}, and ensuring the finite-time backward reachability of a \textit{backup set}. 
This construction yields forward invariance conditions
which can be enforced via a quadratic program (QP).
This approach scales to high-dimensional systems and has showcased strong performance across complex robotic platforms \cite{gurriet_scalable_2020,singletary2021onboard,ko2024backup,janwani_multibackup}.

One of the key advantages of the bCBF method is that expanding the backup set with a pre-verified backup controller yields a tractable optimization problem and desirable properties on the expanded set. 
However, by requiring that the expansion is done with the backup controller,
the possible set of controllers is significantly restricted.
This may yield an expanded set that is overly conservative, potentially resulting in reduced performance. In this work, we provide a method to separate the expanding controller and the backup controller, thus allowing for a much more general technique for set expansion. 
Our contributions are as follows.
\begin{itemize}
    \item 
    We generalize the backup set method by decoupling the set-expanding controller from the backup controller that certifies the invariance of the backup set, and show that the resulting implicit safe set is controlled invariant.
    \item 
    We derive forward invariance conditions for the generalized implicit set and formulate the corresponding
    QP that
    preserves safety guarantees with bounded inputs.
    \item 
    We extend the framework to parameterized controller families, enabling online adaptation of the expansion controller while preserving invariance guarantees.
\end{itemize}
\section{Preliminaries}\label{sec:prelims}

\subsection{Control Barrier Functions}\label{sec:CBF}
Consider the control affine dynamical system
\begin{align}\label{eq:affine-dynamics}
    \dot{\x} = \boldsymbol{f}(\x) + \boldsymbol{g}(\x)\boldsymbol{u}, \nspace{6} \x \in \mathcal{X} \subseteq \R^n, \nspace{6} \boldsymbol{u} \in \mathcal{U} \subseteq \R^m,
\end{align}
where ${\boldsymbol{f}:\mathcal{X} \to \R^n}$ and ${\boldsymbol{g}:\mathcal{X} \to \R^{n \times m}}$ are continuously differentiable,
and the input constraint set $\mathcal{U}$ is a convex polytope.
In this work, we are interested in developing techniques for designing a controller ${\boldsymbol{k}:\mathcal{X} \to \mathcal{U}}$, ${\boldsymbol{u} = \boldsymbol{k}(\x)}$, that guarantees the safety of the closed-loop system corresponding to \eqref{eq:affine-dynamics}.
We frame safety as membership of the state to a set ${\Cs \subseteq \mathcal{X}}$ defined by a continuously differentiable safety function ${h : \mathcal{X} \rightarrow \R}$ as ${\Cs \triangleq \{\x\in\mathcal{X} : h(\x) \geq 0 \}}$.
We view safety in the sense of forward invariance. Formally, a set $\Cs$ is \textit{forward invariant} along the closed-loop system associated with \eqref{eq:affine-dynamics} if ${\x_0 \in \Cs \implies \x(t) \in \Cs, \forall \nspace{1} t \geq 0}$.

\textit{Control barrier functions} (CBFs) provide forward invariance conditions and enable the synthesis of safe controllers.
We call $h$ a CBF for \eqref{eq:affine-dynamics} if ${\sup_{\boldsymbol{u} \in \mathcal{U}} \dot{h}(\x,\boldsymbol{u}) > -\alpha(h(\x))}$ holds for all ${\x \in \Cs}$ with an extended class-$\mathcal{K}$ function $\alpha$.
\begin{theorem}[\hspace{-0.01em}\cite{ames_2017}] \label{thm: cbf}
If $h$ is a CBF for \eqref{eq:affine-dynamics} on $\Cs$, then any locally Lipschitz controller $\boldsymbol{k}:\mathcal{X} \to \mathcal{U}$, $\boldsymbol{u}=\boldsymbol{k}(\x)$ satisfying 
\begin{align} \label{eq: cbf_condition}
    \dot{h}(\x, \boldsymbol{u}) \ge -\alpha(h(\x)),
\end{align}
for all $\x \in \Cs$ renders the set $\Cs$ forward invariant.
\end{theorem}
Given a primary controller ${\boldsymbol{k}_{\rm p} : \mathcal{X} \rightarrow \mathcal{U}}$, safety can be enforced via a quadratic program (QP) subject to the constraint in \eqref{eq: cbf_condition}, to solve for a minimally-invasive safe control signal.
Thus, CBFs are a powerful tool for generating safe controllers, but verifying that a safety function $h$ is actually a CBF is notoriously challenging, especially for
nonlinear systems
with tight input bounds in $\mathcal{U}$.
When input bounds are present, we must consider the notion of controlled invariance. A set $\Cs$ is \textit{controlled invariant} if there exists a controller $\boldsymbol{k} : \mathcal{X} \!\rightarrow\! \mathcal{U}$,  $\boldsymbol{u} \!=\! \boldsymbol{k}(\x)$ rendering $\Cs$ forward invariant for \eqref{eq:affine-dynamics}.

\subsection{Backup Control Barrier Functions}\label{sec:bCBF}

Although finding a CBF is challenging with input bounds, it is often possible
via constructive methods, e.g., by linearizing \eqref{eq:affine-dynamics} about a stabilizable equilibrium
and considering a sublevel set of a Lyapunov function\footnote{This approach is successfully used to generate a backup controller and set for a nonlinear quadrotor in \Cref{ex:quadrotor}.}.
Yet, these techniques often result in a very conservative CBF (with a very small safe set), which can significantly inhibit the performance of CBF-based controllers.
However, as shown in \cite{gurriet_scalable_2020,gurriet_online_2018}, such conservative sets can be \textit{expanded} through a backup controller, resulting in a less conservative definition of safety. We now review this technique, and formalize the notion of set ``expansion''. For this, we require the following definition.
\begin{definition}[Backup Set and Controller] \label{def:backup}
    A continuously differentiable controller ${\ub \!:\! \mathcal{X} \!\to\! \mathcal{U}}$ is called a \textit{backup controller} if it renders a set ${\Cb \!\triangleq\! \{ \x \!\in\! \mathcal{X} \!:\! h_{\rm b}(\x) \!\geq\! 0 \} \!\subseteq\! \Cs}$, called a \textit{backup set}, forward invariant for \eqref{eq:affine-dynamics}. The continuously differentiable function ${h_{\rm b} \!:\! \mathcal{X}\! \rightarrow\! \R}$ is called the \textit{backup function} 
    and has zero as a regular value (i.e., ${\nabla h_{\rm b}(\x) \!\neq\! \mathbf{0}
    , \forall \nspace{2}
    \x \!\in\! \partial \Cb}$). 
\end{definition}
\begin{remark}
By definition, the backup function ${h_{\rm b}}$ is a CBF if $\Cb$ is compact. Constructive methods for obtaining backup controller and set pairs are discussed extensively in \cite{gacsi2025braking}.
\end{remark}

We will henceforth not assume that $h$ is a CBF, but $h$ merely encodes a safety constraint. Because ${h_{\rm b}}$
is a CBF, this could be used directly 
as in \Cref{sec:CBF}
to enforce safety, given that ${\Cb \subseteq \Cs}$. However, as previously mentioned, techniques to construct CBFs often result in conservative sets, yielding overly restrictive safe controllers. The backup set method expands $\Cb$ by considering the evolution of the dynamical system under the backup controller. Consider the \textit{flow} of \eqref{eq:affine-dynamics} under $\ub$, written as $\phinom$, over ${\tau \!\in\! [0,T]}$ with a constant time horizon ${T > 0}$. We can obtain this flow by forward integrating the following differential equation:
\begin{align} \label{eq: f_b}
    \frac{\partial}{\partial \tau}\phinom = \boldsymbol{f}_{\rm b}(\phinom), \quad \boldsymbol{\phi}_{\rm b}(0,\x) = \x,
\end{align}
with the short-hand notation ${\boldsymbol{f}_{\rm b}(\x)\triangleq \boldsymbol{f}(\x) + \boldsymbol{g}(\x)\ub(\x)}$. 

Now consider the expanded set defined using this flow:
\begin{align} \label{def:C_BI}
    \Cbi \triangleq \left\{ \x \in \mathcal{X} \,\middle| 
    \begin{array}{c}
    h(\phinom) \geq 0, \forall \nspace{1} \tau \in [0,T], \\
    h_{\rm b}(\phinomT) \geq 0 \\
    \end{array}
    \!\!\right\}
    \subseteq \Cs.
\end{align}
The set $\Cbi$ contains the states from which the system \eqref{eq:affine-dynamics} can \textit{safely} reach the backup set, $\Cb$, in a finite time $T$, when using the backup controller $\ub$. We emphasize that $\Cbi$ encodes that the flow is contained within $\Cs$
over the entirety of the time horizon, ${\tau \in [0,T]}$. This ensures that ${\Cbi \!\subseteq\! \Cs}$, and thus membership of $\x$ to $\Cbi$ implies membership of $\x$ to $\Cs$. Further, we note that because ${\Cbi \equiv \Cb}$ for ${T = 0}$, and $\Cbi$ grows for ${T > 0}$, $\Cbi$ is indeed an \textit{expansion} of $\Cb$. We now review the main benefit of expanding the backup set in this manner:
$\Cbi$ is controlled invariant by construction.
\begin{lemma}[\!\!\cite{gurriet_scalable_2020}\!{\cite[Lem.~1]{tamasACC_ROM_bCBF}}] \label{lem: CI_controlledInv}
The set $\Cbi$ in \eqref{def:C_BI} is controlled invariant, and the backup controller ${\ub : \mathcal{X}
\rightarrow \mathcal{U}}$ renders $\Cbi$ forward invariant along \eqref{eq:affine-dynamics}, such that ${\x \in \Cbi \implies \phinom \in \Cbi \subseteq \Cs, \forall  \nspace{1} \tau \geq 0}$.
\end{lemma}

\begin{theorem}[\!\!\cite{gurriet_scalable_2020}]
    There exist class-$\mathcal{K}_{\infty}$ functions $\alpha$, $\alpha_{\rm b}$ and a locally Lipschitz controller ${\boldsymbol{k} \!:\! \mathcal{X} \!\rightarrow\! \mathcal{U}, \boldsymbol{u} \!=\! \boldsymbol{k}(\x)}$ satisfying
    \begin{subequations} \label{eq: CI_invariance_cond}
    \begin{align}
        \dot{{h}}(\phinom,\boldsymbol{u}) &\geq -\alpha\big({h}(\phinom)\big), \sspace \forall \tau \in [0,T], \label{eq: htraj_nom}\\
        \dot{{h}}_{\rm b}(\phinomT,\boldsymbol{u}) &\geq -\alpha_{\rm b}\big({h}_{\rm b}(\phinomT)\big), \label{eq: hb_nom}
    \end{align}
    \end{subequations}
    for all
    ${\x \in \Cbi}$. The same controller renders $\Cbi \subseteq \Cs$ forward invariant, and satisfies the input constraints.
\end{theorem}
We are now in a position to enforce the safety of \eqref{eq:affine-dynamics} using the less conservative set, $\Cbi$. Similarly to \Cref{sec:CBF}, it is possible to solve a QP to obtain the safe control signal:
\begin{align*} 
    {\boldsymbol{k}}_{\rm safe}(\x) = \underset{{\boldsymbol{u}} \in {\mathcal{U}}}{\text{argmin}} \mkern9mu &
    \| {\boldsymbol{u}} - {\boldsymbol{k}}_{\rm p}(\x) \| ^{2} \quad 
    \tag{bCBF-QP} \label{eq:bcbf-qp}
    \\
    \text{s.t.  } 
    & \eqref{eq: htraj_nom}, \ \eqref{eq: hb_nom}.
\end{align*}
As written, $\eqref{eq: htraj_nom}$ represents an infinite number of constraints as $\tau$ varies continuously over the interval ${[0,T]}$. Therefore, in practice, this horizon is discretized
and the flow is evaluated at a finite number of points along the horizon.
We assume for the remainder of the paper that for sufficiently 
fine discretization,
satisfaction of the discretized version of $\eqref{eq: htraj_nom}$ implies satisfaction of its continuous counterpart.

Though bCBFs have been used successfully across a wide range of applications, the size of the expanded, implicit safe set $\Cbi$ is inherently tied to the choice of the backup controller. In this work, we challenge this paradigm, and offer a generalization of the standard backup set method.
\section{Generalized Backup Set Method}\label{sec:generalbCBF}

Our main idea is that the controller which expands the backup set $\Cb$ can be \textit{different} from the one which renders $\Cb$ invariant. Specifically, the backup controller does not have to be the one that expands $\Cb$ to obtain an implicit safe set. 

Assume that there exists a continuously differentiable controller ${\boldsymbol{k}_{\rm e} : \mathcal{X} \rightarrow \mathcal{U}}$ which represents the \textit{set-expanding} controller.
Though the controller $\boldsymbol{k}_{\rm e}$ may be a good set expander\footnote{The notion of a ``good'' set expander will be further elaborated upon.},
it does not
necessarily render the backup set forward invariant. Because of this, one cannot directly use this
controller
to construct an implicit safe set as done in \Cref{sec:bCBF} for $\Cbi$ (c.f. \eqref{def:C_BI}). Instead, we can still leverage the properties of the backup controller from \Cref{def:backup}, by considering a switching-based controller, written as
\begin{align}\label{eq:switching_controller}
    \boldsymbol{k}_{\rm s}(\x) = \big(1 - \eta(\x)\big)\kstar(\x) + \eta(\x)\ub(\x),
\end{align}
with a switching function ${\eta : \mathcal{X} \rightarrow [0,1]}$. Notice that by construction, ${\boldsymbol{k}_{\rm s}(\x) \in \mathcal{U}}$ for all ${\x \in \mathcal{X}}$ if ${\ub(\x) \in \mathcal{U}}$, ${\boldsymbol{k}_{\rm e}(\x) \in \mathcal{U}}$, and the input constraint set $\mathcal{U}$ is convex. We now restrict the choice of $\eta$ such that $\boldsymbol{k}_{\rm s}$ has desirable properties.
\begin{assumption} \label{ass:switchingFun}
    The switching function ${\eta : \mathcal{X} \rightarrow [0,1]}$ is continuously differentiable, and satisfies $\eta(\x) = 1, \nspace{2}\forall \nspace{1} \x \in \Cb$.
\end{assumption}
An example of many possible switching functions satisfying \Cref{ass:switchingFun} is ${\eta(\x) = \Gamma(h_{\rm b}(\x))}$ where ${\Gamma : \R \rightarrow [0,1]}$
is a
continuously differentiable
step function given by
\begin{align} \label{eq:smoothstep}
    \Gamma(z) \triangleq
    \begin{cases}
        0 & {\rm if}\ z \leq -\varepsilon, \\
        3(\frac{z + \varepsilon}{\varepsilon})^2 - 2(\frac{z + \varepsilon}{\varepsilon})^3 & {\rm if}\ -\varepsilon < z < 0,\\
        1 & {\rm if}\ z \geq 0,
    \end{cases}
\end{align}
and ${\varepsilon > 0}$ governs the switching region.
\begin{lemma} \label{lemma:CB_invariant_switch}
    The switched controller $\boldsymbol{k}_{\rm s}$ in \eqref{eq:switching_controller} renders the backup set $\Cb$ satisfying \Cref{def:backup} forward invariant along \eqref{eq:affine-dynamics} for any switching function
    $\eta$
    satisfying \Cref{ass:switchingFun}.
    \begin{proof}
       By \Cref{ass:switchingFun}, for all ${\x \in \Cb}$ we have that ${\boldsymbol{k}_{\rm s}(\x) = \ub(\x)}$ as per \eqref{eq:switching_controller}, and by \Cref{def:backup}, the backup controller $\ub$ renders $\Cb$ forward invariant.
    \end{proof}
\end{lemma}
\begin{remark}
   For any $\eta$ satisfying \Cref{ass:switchingFun}, the switched controller $\boldsymbol{k}_{\rm s}$ itself is a valid backup controller for $\Cb$.
\end{remark}
Using \eqref{eq:switching_controller}, the switched dynamical system is
\begin{align}\label{eq: f_cl_switched}
    \dot{\boldsymbol{x}} = \boldsymbol{f}_{\rm s}(\boldsymbol{x}) \triangleq 
    \boldsymbol{f}(\x) + \boldsymbol{g}(\x)\boldsymbol{k}_{\rm s}(\x).
\end{align}
With the above definitions, the expanded safe set is
\begin{align} \label{def:C_BI_new}
    \Cbistar \triangleq \left\{\!
    \x \in \mathcal{X}
    \middle|
    \begin{array}{c}
    h({\boldsymbol{\phi}_{\rm s}}(\tau,\x)) \geq 0, \forall \nspace{1} \tau \in [0,T], \\
    h_{\rm b}({\boldsymbol{\phi}_{\rm s}}(T,\x)) \geq 0 \\
    \end{array} \!
    \right\}.
\end{align}
Here, the switched flow ${{\boldsymbol{\phi}_{\rm s}}(\tau,\x)}$ is the evolution of \eqref{eq: f_cl_switched} over ${\tau\!\in\![0,T]}$ for a fixed, finite horizon ${T \!>\!0}$ starting at state $\x$:
\begin{align} \label{eq: flow_switched}
    \frac{\partial}{\partial \tau}{\boldsymbol{\phi}_{\rm s}}(\tau,\x) = \boldsymbol{f}_{\rm s}({\boldsymbol{\phi}_{\rm s}}(\tau,\x)), \sspace \phins{0}{\x} = \boldsymbol{x}.
\end{align}
\begin{remark}
A possible switching function that would render $\Cb$ forward invariant could be ${\eta(\x) = {1}}$ if ${h_{\rm b}(\x) \geq 0}$ and ${\eta(\x) = {0}}$ otherwise.
However, the discontinuity of this switch would make the controller \eqref{eq:switching_controller}
and
the switched system \eqref{eq: f_cl_switched} non-smooth.
Restricting $\eta$ to be continuously differentiable admits global solutions ${\boldsymbol{\phi}_{\rm s}(\tau,\x)}$ and well-behaved gradients.
\end{remark}

We now state a powerful result for
the expanded set $\Cbistar$.

\begin{lemma}
    \label{lemma:C_I_star_controlled_invariant}
    The set ${\mathcal{C}^{\star}_{\rm I}}$ in \eqref{def:C_BI_new} is controlled invariant, and the switched controller ${\boldsymbol{k}_{\rm s}}$ renders ${\mathcal{C}^{\star}_{\rm I}}$ forward invariant along \eqref{eq: f_cl_switched}, such that ${\x \in \Cbistar \implies {\boldsymbol{\phi}_{\rm s}}(\tau,\x) \in \Cbistar \subseteq \Cs, \forall \tau \geq 0}$.
    \begin{proof}
    From the definition of $\Cbistar$, we have $    \boldsymbol{x} \in \Cbistar \implies \boldsymbol{\phi}_{\rm s}(T, \boldsymbol{x}) \in \Cb$,
    and with \Cref{lemma:CB_invariant_switch}, 
    \begin{align} \label{eq: l1_1}
        \boldsymbol{x} \in \Cbistar \implies \boldsymbol{\phi}_{\rm s}(\tau, \boldsymbol{x}) \in \Cb \subseteq \Cs, \forall \nspace{1} \tau \geq T.
    \end{align}
    The flow is recursive and thus for any ${\boldsymbol{x} \in \mathbb{R}^n}$ and ${\tau, \vartheta \geq 0}$, ${\boldsymbol{\phi}_{\rm s} (\tau \!+\! \vartheta, \boldsymbol{x}) \!=\! \boldsymbol{\phi}_{\rm s} (\tau, \boldsymbol{\phi}_{\rm s} (\vartheta, \boldsymbol{x}))}$.
    Using this property and \eqref{eq: l1_1},
    \begin{align} \label{eq: l1_Tb}
        \boldsymbol{x} \in \Cbistar \implies \boldsymbol{\phi}_{\rm s} (T, \boldsymbol{\phi}_{\rm s} (\vartheta, \boldsymbol{x})) \in \Cb, \forall \nspace{1}\vartheta \geq 0.
    \end{align}
    From \eqref{eq: l1_1} and \eqref{def:C_BI_new}, ${\boldsymbol{x} \in \Cbistar \implies \boldsymbol{\phi}_{\rm s} (\tau, \boldsymbol{x}) \in \Cs, \forall \nspace{1} \tau \geq 0}$.
    Using the recursive property once more
    \begin{align} \label{eq: l1_done}
        \nspace{-3} \boldsymbol{x} \in \Cbistar \nspace{-6} \implies \nspace{-6}\boldsymbol{\phi}_{\rm s} (\tau, \boldsymbol{\phi}_{\rm s} (\vartheta, \boldsymbol{x})) \nspace{-2} \in \nspace{-2} \Cs, \forall \nspace{1} \tau \nspace{-1}\in \nspace{-1}[0, T], \forall \vartheta \nspace{-2}\geq \nspace{-2}0.
    \end{align}
    Definition \eqref{def:C_BI_new} with \eqref{eq: l1_Tb} and \eqref{eq: l1_done} completes the proof.
\end{proof}    
\end{lemma}

Now that $\Cbistar$ has been shown to be controlled invariant, we derive its forward invariance conditions which may be used to solve for a safe control signal. From the definition of $\Cbistar$, we require that for some class-$\mathcal{K}_{\infty}$ functions $\alpha$, $\alpha_{\rm b}$
\begin{align}
\begin{aligned} \label{eq:invariance_Cbi}
\dot{{h}}(\boldsymbol{\phi}_{\rm s}(\tau, \boldsymbol{x}),\boldsymbol{u}) &\geq -\alpha\big({h}(\boldsymbol{\phi}_{\rm s}(\tau, \boldsymbol{x}))\big), \sspace \forall \tau \in [0,T], \\
\dot{{h}}_{\rm b}(\boldsymbol{\phi}_{\rm s}(T, \boldsymbol{x}),\boldsymbol{u}) &\geq -\alpha_{\rm b}\big({h}_{\rm b}(\boldsymbol{\phi}_{\rm s}(T, \boldsymbol{x}))\big),
\end{aligned}
\end{align}
holds.
From the chain rule we have
\begin{subequations}
\begin{align}
    \dot{{h}}(\boldsymbol{\phi}_{\rm s}(\tau, \boldsymbol{x}),\boldsymbol{u}) &= \nabla {h}(\boldsymbol{\phi}_{\rm s}(\tau, \boldsymbol{x})){\boldsymbol{\Phi}}(\tau,\x)\dot{\x}, \label{eq:h_hat_dot_1} \\
    \dot{{h}}_{\rm b}(\boldsymbol{\phi}_{\rm s}(T, \boldsymbol{x}),\boldsymbol{u}) &= \nabla {h}_{\rm b}(\boldsymbol{\phi}_{\rm s}(T, \boldsymbol{x})){\boldsymbol{\Phi}}(T,\x)\dot{\x},
    \label{eq:hb_hat_dot_1}
\end{align}
\end{subequations}
with ${\dot{\x} = {\boldsymbol{f}}(\x) + {\boldsymbol{g}}(\x) {\boldsymbol{u}}}$.
Here, ${\boldsymbol{\Phi}(\tau,\x) \triangleq \frac{\partial \boldsymbol{\phi}_{\rm s}(\tau, \boldsymbol{x})}{\partial \x}}$ is the sensitivity of the switched flow to
$\x$, the initial condition for \eqref{eq: flow_switched}. The sensitivity matrix is the solution to 
\begin{align}\label{eq:STM_gen}
    \frac{\partial}{\partial \tau}{\boldsymbol{\Phi}}(\tau,\x)  =  {\boldsymbol{F}}_{\rm s}(\boldsymbol{\phi}_{\rm s}(\tau, \boldsymbol{x})){\boldsymbol{\Phi}}(\tau,\x), \nspace{12} {\boldsymbol{\Phi}}(0,\x) =  \boldsymbol{I}_{n},
\end{align}
where ${{\boldsymbol{F}}_{\rm s}(\x) \triangleq \frac{\partial {\boldsymbol{f}}_{\rm s}}{\partial \x}(\x)}$ is the closed-loop Jacobian of \eqref{eq: f_cl_switched}.

We now present the main result.
\begin{theorem} \label{thm:Invariance_Cbistar}
    Any locally Lipschitz controller ${{\boldsymbol{k}}:{\mathcal{X}}\rightarrow{\mathcal{U}}}$, ${\boldsymbol{u}} = {\boldsymbol{k}}(\x)$ satisfying 
   \begin{subequations} \label{eq:invariance_Cbi_2}
       \begin{align}
       &\begin{aligned} \label{eq:h_hat_invariance_1}
           \nabla {h}(\boldsymbol{\phi}_{\rm s}(\tau, \boldsymbol{x}))&{\boldsymbol{\Phi}}(\tau,\x)\big( {\boldsymbol{f}}(\x) + {\boldsymbol{g}}(\x) {\boldsymbol{u}}\big) \geq \\ &-\alpha\big({h}(\boldsymbol{\phi}_{\rm s}(\tau, \boldsymbol{x}))\big), \sspace \forall \tau \in [0,T],
       \end{aligned} \\
       &\begin{aligned} \label{eq:hb_hat_invariance_1}
           \nabla {h}_{\rm b}(\boldsymbol{\phi}_{\rm s}(T, \boldsymbol{x}))&{\boldsymbol{\Phi}}(T,\x)\big( {\boldsymbol{f}}(\x) + {\boldsymbol{g}}(\x) {\boldsymbol{u}}\big) \geq \\ &-\alpha_{\rm b}\big({h}_{\rm b}(\boldsymbol{\phi}_{\rm s}(T, \boldsymbol{x}))\big),
       \end{aligned}
       \end{align}
   \end{subequations}
   for all
   ${\x \!\in\! \Cbistar}$, renders ${\Cbistar \!\subseteq\! \Cs}$ forward invariant for $\eqref{eq:affine-dynamics}$.
   \begin{proof}
       By applying \Cref{thm: cbf} to system \eqref{eq:affine-dynamics}, we have that the conditions in \eqref{eq:invariance_Cbi_2} imply that for any ${\x_0 \in \Cbistar}$, ${\x(t) \in \Cbistar}$ for all $t \geq 0$. By definition, $\Cbistar \subseteq \Cs$.
   \end{proof}
\end{theorem}
We use \Cref{thm:Invariance_Cbistar} to obtain a point-wise optimal safe controller via the proposed \textit{Generalized Backup CBF} method:
\begin{align*} 
    {\boldsymbol{k}}_{\rm safe}(\x) = \underset{{\boldsymbol{u}} \in {\mathcal{U}}}{\text{argmin}} \mkern9mu &
    \| {\boldsymbol{u}} - {\boldsymbol{k}}_{\rm p}(\x) \| ^{2} \quad 
    \tag{{GB}-QP} \label{eq:gen-qp}
    \\
    \text{s.t.  } 
    & \eqref{eq:h_hat_invariance_1}, \ \eqref{eq:hb_hat_invariance_1}.
\end{align*}
Further,
by \Cref{lemma:C_I_star_controlled_invariant}, there exists at least one controller (namely $\boldsymbol{k}_{\rm s}$) which renders $\Cbistar$ forward invariant, and satisfies the input constraints imposed by $\mathcal{U}$. 

Having established the theoretical basis for generalizing the backup set method, we may now direct our attention to an example to answer the question: \textit{``Why should the set-expanding controller be different from the backup controller?''}

\begin{figure}
    \centering
    \hspace{-.3cm}
    \includegraphics[width=1\linewidth]{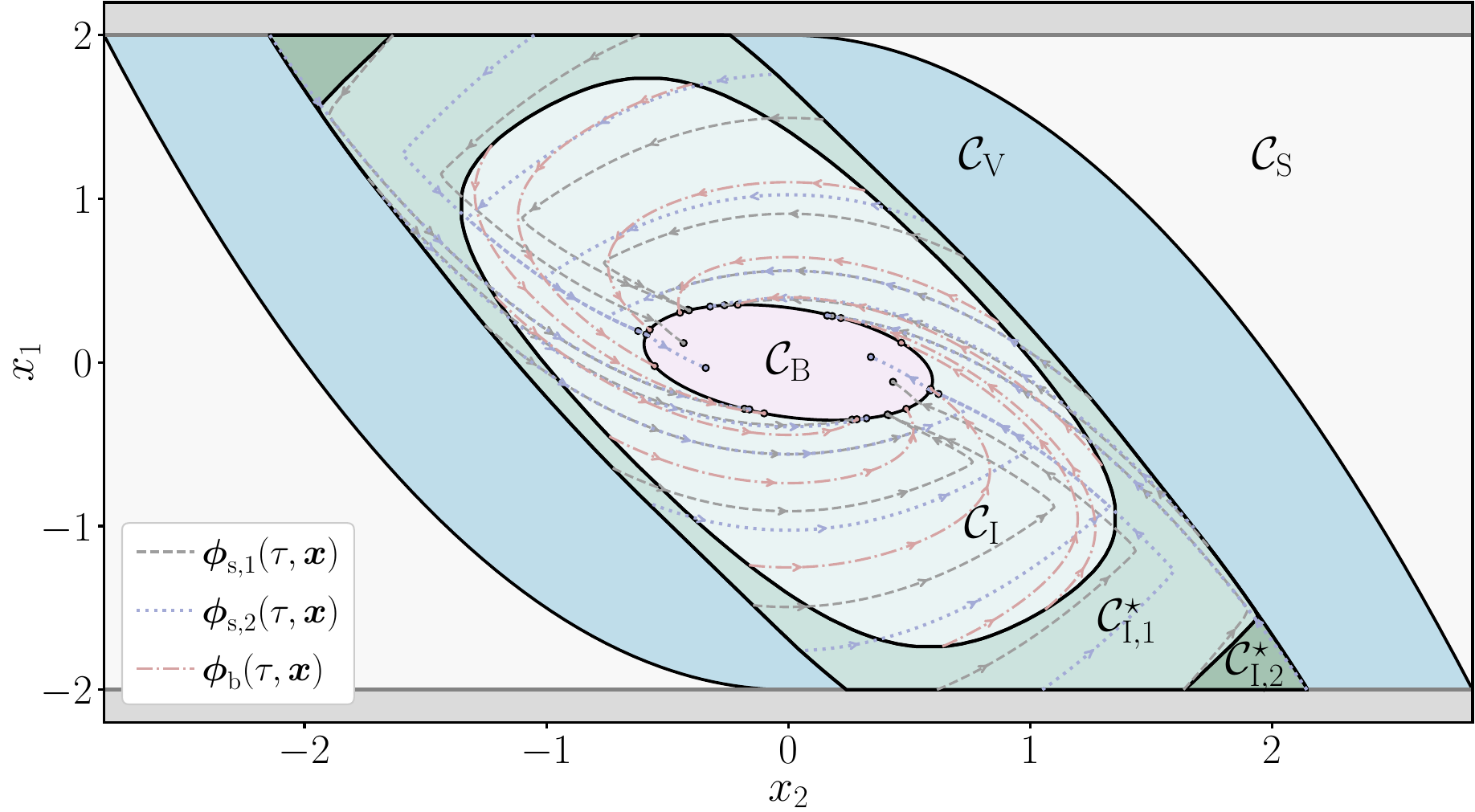}
    \vspace{-.05cm}
    \caption{Controlled invariant sets for \eqref{eq: db_int} with three different expansion strategies.
    $\Cbi$ is constructed using the closed-loop flow for the backup controller $\ub$, whilst $\mathcal{C}^\star_{\rm I,1}$ and $\mathcal{C}^\star_{\rm I,2}$ use the switching controller $\boldsymbol{k}_{\rm s}$ in \eqref{eq:switching_controller}, which allows for a more flexible set expansion policy. The flows denoted by $\boldsymbol{\phi}_{\rm b}$, $\boldsymbol{\phi}_{\rm s,1}$, and $\boldsymbol{\phi}_{\rm s,2}$ represent the evolution of \eqref{eq: db_int} under each of the controllers. The viability kernel, denoted by $\mathcal{C}_{\rm V}$, is plotted for comparison.
    }
    \label{fig:CI_dbint}
    \vskip - 4mm
\end{figure}

\begin{casestudy}[Double Integrator] \label{ex:db_int_motivation}
Consider the system
\begin{align} \label{eq: db_int}
    \dot{\boldsymbol{x}} = \boldsymbol{A}\x + \boldsymbol{B}\boldsymbol{u}, \nspace{10} \boldsymbol{A} = \begin{bmatrix}
        0 & 1 \\
        0 & 0 
    \end{bmatrix}, \nspace{5}\boldsymbol{B} =\begin{bmatrix}
        0 \\
        1
    \end{bmatrix},
\end{align}
with state ${\boldsymbol{x} =
[x_1 ~ x_2]^\top \in \mathbb{R}^2}$, input ${u \in \mathcal{U} = [-1, 1]}$, and constraint set $\Cs \!\triangleq\! \{\boldsymbol{x} \!\in\! \mathbb{R}^2 \!:\! h(\x) \!=\! x^2_{\rm max}\! -\! x_1^2 \!\geq \!0\}$.
Let the backup set be
$\mathcal{C}_{\rm B} \triangleq \{\boldsymbol{x} \!\in\! \mathbb{R}^2 : h_{\rm b}(\x) = \rho - \x^\top \boldsymbol{P} \x \geq 0  \}$,
for a positive definite matrix $\boldsymbol{P}$ and a ${\rho >0}$ such that ${\Cb \subseteq \Cs}$.
Consider the backup controller with gain $\boldsymbol{K}$ as
\begin{align} \label{eq: kb_dbint}
  \ub(\x) = {\rm sat}(-\boldsymbol{K} \x),
\end{align}
where ${\rm sat}$ denotes any smooth saturation function.
To ensure that
$\ub$ does not saturate inside $\Cb$,
we require a condition\footnote{By \cite[Eq. (B4)]{khalil_nonlinear_2015} we require ${\max_{\x^\top \boldsymbol{P} \x \leq \rho} |\boldsymbol{K}\x| \!=\! \sqrt{\rho} \|{\boldsymbol{K} \boldsymbol{P}^{-1/2} \|} \!\leq\! 1}$.} 
on 
$\boldsymbol{K}$
so that
${\ub(\x) \!=\! -\boldsymbol{K} \x}$ for all ${\x \in \Cb}$.
Then, 
this controller renders $\Cb$ forward invariant 
if 
$\boldsymbol{K}$
satisfies
${\boldsymbol{A}_{\rm cl}^\top \boldsymbol{P} \!+\! \boldsymbol{P}\boldsymbol{A}_{\rm cl} \!=\! - \boldsymbol{Q}}$ where ${\boldsymbol{A}_{\rm cl} \!\triangleq \!\boldsymbol{A}\! -\! \boldsymbol{B}\boldsymbol{K}}$
with a positive definite matrix $\boldsymbol{Q}$.
Thus, to render $\Cb$ forward invariant while satisfying input limits for $\ub$, the choice of the gain $\boldsymbol{K}$ is restricted.
This restriction may negatively impact the size of the set $\Cbi$ when the same controller $\ub$ is used to expand $\Cb$.

Now, consider the expanding controller
\begin{align} \label{eq: k_star_dbint}
  \boldsymbol{k}_{\rm e,1}(\x) = {\rm sat}(-\boldsymbol{K}_{\rm e}\x),
\end{align}
with a gain $\boldsymbol{K}_{\rm e}$.
Because we do not have to verify that this controller renders $\Cb$ forward invariant, and because \eqref{eq: k_star_dbint} is saturated to satisfy input constraints, the gain ${\boldsymbol{K}_{\rm e}}$ can be made arbitrarily large.
Choosing a larger gain may facilitate a larger expansion $\Cbistar$ of the backup set $\Cb$.
Alternatively, to further improve expansion, we may also consider a smooth approximation of the time-optimal controller described in \cite[Ch. 3.9]{bryson1975} which reaches the origin in minimum time:
\begin{align} \label{eq: k_opt_dbint}
  \boldsymbol{k}_{\rm e,2}(\x) = - {\rm sat}\big((x_2^2 \nspace{2}{\rm sat}(x_2/\beta) + 2x_1)/\beta \big),
\end{align}
where ${\beta > 0}$ is a smoothing parameter.
\Cref{fig:CI_dbint} compares\footnote{The simulation uses the constants ${\boldsymbol{K} = [2, 1.6]}$, ${\boldsymbol{K}_{\rm e} = 30\nspace{1}\boldsymbol{K}}$, ${\rho = 0.15}$, ${\boldsymbol{Q} = \boldsymbol{I}_2}$, ${\beta = 1\times10^{-9}}$, ${\varepsilon = 1\times10^{-3}}$, ${T = 2}$.} the size of $\Cbi$ defined in \eqref{def:C_BI} with $\mathcal{C}^\star_{\rm I,1}$ and $\mathcal{C}^\star_{\rm I,2}$ defined in \eqref{def:C_BI_new} using $\boldsymbol{k}_{\rm s}$ in \eqref{eq:switching_controller} with the expanding controllers $\boldsymbol{k}_{\rm e,1}$ and $\boldsymbol{k}_{\rm e,2}$ respectively, and the switching function in \eqref{eq:smoothstep}. Clearly,  ${\Cbi \subseteq \mathcal{C}^\star_{\rm I,1}, \mathcal{C}^\star_{\rm I,2}}$, illustrating that the generalized backup CBF method yields significantly larger safe sets in this example.
\end{casestudy}
By decoupling the backup and expanding controllers, one may construct a larger controlled invariant set (see \Cref{fig:CI_dbint}), leading to
less conservative safe controllers.
Further, by establishing the necessary machinery for this decoupling, we allow for
a very broad class of expanding controllers.
Note, however, that
not all controllers may be \textit{useful} expanders\footnote{For an ill-informed choice of $\boldsymbol{k}_{\rm e}$, we have $\Cbistar\equiv\Cb$ in the worst case.}.
\begin{remark}[Optimal Expanding Controller] \label{rem:OCP}
We recall the optimal control problem\footnote{${\mathcal{X}_{\rm c}}$ and  ${\mathcal{U}_{\rm c}}$ are the set of piecewise continuous states and control signals.}
from \cite{gurriet_online_2018}:
\begin{align} \label{eq:OCP}
\begin{aligned}
    \boldsymbol{k}^*_{\rm e}&(\x) \triangleq 
    \argmax_{\boldsymbol{\phi} \in \mathcal{X}_{\rm c}, \boldsymbol{u} \in \mathcal{U}_{\rm c}} \quad h_{\rm b}(\boldsymbol{\phi}(T,\x)) \\
    {\rm s.t.} \nspace{10} 
    & \frac{\partial}{\partial \tau}\boldsymbol{\phi}(\tau,\x) \!=\! \boldsymbol{f}(\boldsymbol{\phi}(\tau,\x)) \!+\! \boldsymbol{g}(\boldsymbol{\phi}(\tau,\x))\boldsymbol{u}(\tau), \forall \tau \!\in\! [0,T], \\
    &\boldsymbol{u}(\tau) \in \mathcal{U}, \sspace \forall \tau \in [0,T],
    \raisetag{4.9\baselineskip} 
    \\
    & \boldsymbol{\phi}(0,\x) = \x.
\end{aligned}
\end{align}
The controller ${\boldsymbol{k}_{\rm e}^*}$ that solves \eqref{eq:OCP} is an optimal expanding control policy
ensuring that
the terminal state, 
$\boldsymbol{\phi}(T,\x)$,
lies as far as possible inside the backup set. Though solving \eqref{eq:OCP} is challenging for nonlinear systems and arbitrary $h_{\rm b}$, minimum-time-to-origin optimal control problems may often be useful surrogates to \eqref{eq:OCP} (see e.g., \eqref{eq: k_opt_dbint}).
Next, we discuss how the cost function in \eqref{eq:OCP} may inform adaptation.
\end{remark}

\begin{remark}
    If ${\boldsymbol{k}_{\rm e}(\x) \!=\! \ub(\x)}$, then ${\Cbistar\!\equiv \!\Cbi}$ and the \eqref{eq:gen-qp} recovers the \eqref{eq:bcbf-qp} as a special case.
\end{remark}
\section{Adaptive Set Expansion} \label{sec:adaptive}

In this section, we consider yet another layer of generalization: we explore expanders that are not fixed, but change over time through adaptation.
Suppose now that the controller $\kstar$ is parameterized\footnote{For example, one could write $\kstar$ in \eqref{eq: k_star_dbint} as ${\kstar(\x,\boldsymbol{\theta}) = {\rm sat}(-\boldsymbol{\theta}^\top \x)}$.} by ${\boldsymbol{\theta} \in \R^{p}}$ and written as $\kstar(\x,\boldsymbol{\theta})$.
Considering the parameters ${\boldsymbol{\theta}}$ as auxiliary states, the full augmented state is written as ${\hat{\boldsymbol{x}} = [\x^\top \nspace{5}\boldsymbol{\theta}^\top]^\top\!\!\in \widehat{\mathcal{X}} \triangleq \mathcal{X} \times \R^{p},}$
and thus the augmented dynamics for \eqref{eq:affine-dynamics} can be written as 
\begin{align} \label{eq:affine-dynamics_aux}
    \dot{\xaux} = 
    \underbrace{
    \begin{bmatrix}
        \boldsymbol{f}(\x) \\
        \boldsymbol{0}_{p \times 1}
    \end{bmatrix}}_{\triangleq \nspace{1}\widehat{\boldsymbol{f}}(\xaux)}
    + 
    \underbrace{
    \begin{bmatrix}
        \boldsymbol{g}(\x) & \boldsymbol{0}_{n \times p} \\
        \boldsymbol{0}_{p \times m} & \boldsymbol{I}_{p}
    \end{bmatrix}}_{\triangleq \nspace{1}\widehat{\boldsymbol{g}}(\xaux)}
    \hat{\boldsymbol{u}},
\end{align}
where ${\hat{\boldsymbol{u}} \triangleq [\boldsymbol{u}^\top \nspace{5}\boldsymbol{u}_{\theta}^\top]^\top \!\!\in \widehat{\mathcal{U}} = \mathcal{U} \times \R^p}$ and ${\boldsymbol{u}_{\theta} \in \R^{p}}$ is an auxiliary input which modifies the parameter $\boldsymbol{\theta}$.
We now define the \textit{augmented} backup controller and expanding controller:
\begin{align*}
    \hat{\boldsymbol{k}}_{\rm b}(\xaux) \!\triangleq \!  
    \begin{bmatrix}
        \boldsymbol{k}_{\rm b}(\x)^\top \nspace{10}
        \boldsymbol{0}_{1 \times p}
    \end{bmatrix}^\top\nspace{-5}, 
    \sspace
    \hat{\boldsymbol{k}}_{\rm e}(\xaux)
    \!\triangleq\! 
    \begin{bmatrix}
        \boldsymbol{k}_{\rm e}(\x,\boldsymbol{\theta})^\top \nspace{10}
        \boldsymbol{0}_{1 \times p}
    \end{bmatrix}^\top\nspace{-6}.
\end{align*}
Notice that each controller keeps $\boldsymbol{\theta}$ fixed. Analogously to \Cref{sec:generalbCBF}, we write the switched augmented dynamics as 
\begin{align}\label{eq: f_s_aux}
    \dot{\xaux} = \widehat{\boldsymbol{f}}_{\rm s}(\xaux) \triangleq 
    \widehat{\boldsymbol{f}}(\xaux) + \widehat{\boldsymbol{g}}(\xaux)\hat{\boldsymbol{k}}_{\rm s}(\xaux),
\end{align}
for the switched controller 
\begin{align}\label{eq:switching_controller_aux}
    \!\hat{\boldsymbol{k}}_{\rm s}(\xaux) \!=\!
    \begin{bmatrix}
        \big(1 \!-\! \eta(\x)\big)\kstar(\x,\boldsymbol{\theta})^\top \!\!+\! \eta(\x)\ub(\x)^\top \nspace{10}
        \boldsymbol{0}_{1 \times p}
    \end{bmatrix}^\top\nspace{-6},
\end{align}
with ${\eta : \mathcal{X} \rightarrow [0,1]}$. 
Then, the augmented switched flow is
\begin{align} \label{eq: auxFlow}
    \frac{\partial}{\partial \tau}{\boldsymbol{\widehat{\phi}}_{\rm s}}(\tau,\xaux) = \widehat{\boldsymbol{f}}_{\rm s}(\phiaux), \sspace \phinbaux{0}{\xaux} = \xaux,
\end{align}
over the interval ${\tau\in[0,T]}$ for a horizon ${T > 0}$.
We redefine the safe and backup sets in terms of the augmented state $\xaux$:
\begin{align*}
    \widehat{\mathcal{C}}_{\rm S} \triangleq \{\xaux \in \widehat{\mathcal{X}} : \hat{h}(\xaux) \ge 0\}, \sspace
    \widehat{\mathcal{C}}_{\rm B} \triangleq \{\xaux \in \widehat{\mathcal{X}} : \hat{h}_{\rm b}(\xaux) \ge 0\},
\end{align*}
where
${\hat{h}(\xaux) \triangleq h(\x)}$ and
${\hat{h}_{\rm b}(\xaux) \triangleq h_{\rm b}(\x)}$.
We now present a result analogous to \Cref{lemma:CB_invariant_switch} for the augmented system.

\begin{lemma} \label{lemma:CB_hat_invariant_switch}
    The controller $\hat{\boldsymbol{k}}_{\rm s}$ in \eqref{eq:switching_controller_aux} renders $\Cbaux$ forward invariant along \eqref{eq:affine-dynamics_aux} for any ${\eta}$ satisfying \Cref{ass:switchingFun}.
    \begin{proof}
    Observe that in \eqref{eq:affine-dynamics_aux}, the dynamics of $\x$ are decoupled from those of $\boldsymbol{\theta}$, and that ${\xaux \in \Cbaux}$ if and only if ${\x \in \Cb}$, for any ${\boldsymbol{\theta}\in \mathbb{R}^p}$.
    Therefore, $\hat{\boldsymbol{k}}_{\rm s}$ rendering $\Cbaux$ invariant along \eqref{eq:affine-dynamics_aux} is equivalent to $\boldsymbol{k}_{\rm s}$ rendering $\Cb$ invariant along \eqref{eq:affine-dynamics}. Thus, by \Cref{lemma:CB_invariant_switch},
    $\Cbaux$ is invariant under $\hat{\boldsymbol{k}}_{\rm s}$.
    \end{proof}
\end{lemma}
We define the new implicit safe set ${\Cbiaux \!\subseteq\! \widehat{\mathcal{C}}_{\rm S}}$ based on $\xaux$:
\begin{align} \label{def:C_BI_aux}
    \Cbiaux \triangleq \left\{ \xaux \in \widehat{\mathcal{X}} \,\middle|\, 
    \begin{array}{c}
    \hat{h}(\phiaux) \geq 0, \forall \nspace{1} \tau \in [0,T], \\
    \hat{h}_{\rm b}(\phiauxT) \geq 0 \\
    \end{array}
    \right\}.
\end{align}
This set encodes all states $\xaux$ for which the $\x$ components of $\phiaux$ stay within $\Cs$ and reach $\Cb$ when ${\tau = T}$. 
We now show that
this new implicit safe set is controlled invariant.

\begin{lemma}
    The set $\Cbiaux$ is controlled invariant and rendered forward invariant along \eqref{eq:affine-dynamics_aux} by 
    $\hat{\boldsymbol{k}}_{\rm s}$.
    \begin{proof}
        The proof follows that of \Cref{lemma:C_I_star_controlled_invariant}.
        Due to ${\xaux \!\in\! \Cbiaux}$ $\!\implies\! \phiauxT \!\in\! \Cbaux$,
        \Cref{lemma:CB_hat_invariant_switch}, and the flow group property, 
        \begin{align} \label{eq:intermed_CIstar_1}
            \xaux \in \Cbiaux \implies \widehat{\boldsymbol{\phi}}_{\rm s} (T, \widehat{\boldsymbol{\phi}}_{\rm s} (\vartheta, \xaux)) \in \Cbaux, \forall \nspace{1}\vartheta \geq 0.
        \end{align}
        Based on $\eqref{def:C_BI_aux}$, the flow group property, and since ${\Cbaux \subseteq \Csaux}$,
        \begin{align} \label{eq:intermed_CIstar_2}
            \nspace{-3} \xaux \in \Cbiaux \nspace{-6} \implies \nspace{-6}\widehat{\boldsymbol{\phi}}_{\rm s} (\tau, \widehat{\boldsymbol{\phi}}_{\rm s} (\vartheta, \xaux)) \nspace{-2} \in \nspace{-2} \Csaux, \forall \nspace{1} \tau \nspace{-1}\in \nspace{-1}[0, T], \forall \vartheta \nspace{-2}\geq \nspace{-2}0.          
        \end{align}
        Thus \eqref{eq:intermed_CIstar_1}, \eqref{eq:intermed_CIstar_2} with definition \eqref{def:C_BI_aux} complete the proof.
    \end{proof}
\end{lemma}
Now that the controlled invariance of $\Cbiaux$ has been established, we are in a position to derive the forward invariance conditions for this set. From definition \eqref{def:C_BI_aux}, this requires
\begin{subequations}  \label{eq:invariance_Cbiaux}
\begin{align}
    \nabla \hat{h}(\phiaux)\widehat{\boldsymbol{\Phi}}(\tau,\xaux)\dot{\xaux}  &\geq -\alpha\big(\hat{h}(\phiaux)\big),
    \label{eq:h_hat_dot} \\
    \nabla \hat{h}_{\rm b}(\phiauxT)\widehat{\boldsymbol{\Phi}}(T,\xaux)\dot{\xaux} &\geq -\alpha_{\rm b}\big(\hat{h}_{\rm b}(\phiauxT)\big), \label{eq:hb_hat_dot}
\end{align}
\end{subequations}
for all ${\tau \in [0,T]}$ and some class-$\mathcal{K}_{\infty}$ functions ${\alpha, \alpha_{\rm b}}$. 
The term ${\widehat{\boldsymbol{\Phi}}(\tau,\xaux) \triangleq \frac{\partial \phiaux}{\partial \xaux}}$ is the sensitivity of the flow with respect to changes in $\xaux$, and is calculated similarly to \eqref{eq:STM_gen}.

\begin{theorem}\label{thm:Invariance_Cbiaux}
   Any locally Lipschitz controller ${\hat{\boldsymbol{k}}:\widehat{\mathcal{X}}\rightarrow\widehat{\mathcal{U}}}$, $\hat{\boldsymbol{u}} = \hat{\boldsymbol{k}}(\xaux)$ satisfying
   \begin{subequations} \label{eq:invariance_Cbiaux_2}
       \begin{align}
       &\begin{aligned} \label{eq:h_hat_invariance}
           \nabla \hat{h}(\phiaux)&\widehat{\boldsymbol{\Phi}}(\tau,\xaux)\big( \widehat{\boldsymbol{f}}(\xaux) + \widehat{\boldsymbol{g}}(\xaux) \hat{\boldsymbol{u}}\big) \geq \\ &-\alpha\big(\hat{h}(\phiaux)\big), \sspace \forall \tau \in [0,T],
       \end{aligned} \\
       &\begin{aligned} \label{eq:hb_hat_invariance}
           \nabla \hat{h}_{\rm b}(\phiauxT)&\widehat{\boldsymbol{\Phi}}(T,\xaux)\big( \widehat{\boldsymbol{f}}(\xaux) + \widehat{\boldsymbol{g}}(\xaux) \hat{\boldsymbol{u}}\big) \geq \\ &-\alpha_{\rm b}\big(\hat{h}_{\rm b}(\phiauxT)\big),
       \end{aligned}
       \end{align}
   \end{subequations}
   for all
   ${\xaux \in \Cbiaux}$, renders $\Cbiaux$ forward invariant for $\eqref{eq:affine-dynamics_aux}$, and thus ensures that ${\xaux_0 \in \Cbiaux \implies \x(t) \in \Cs, \forall \nspace{1} t \geq 0}$.
   \begin{proof}
       By applying \Cref{thm: cbf} to system \eqref{eq:affine-dynamics_aux}, the conditions in \eqref{eq:invariance_Cbiaux_2} imply that ${\xaux(t) \in \Cbiaux}$ for all ${t \geq 0}$ if ${\xaux_0 \in \Cbiaux}$. Since $\Cbiaux \subseteq \Csaux$, and because ${\xaux \in \Csaux}$ if and only if $\x \in \Cs$, for any ${\boldsymbol{\theta}\in \mathbb{R}^p}$, we have that $\x(t) \in \Cs, \forall \nspace{2} t \geq 0$.
   \end{proof}
\end{theorem}
Using \Cref{thm:Invariance_Cbiaux}, we design an \textit{Adaptive Generalized Backup CBF} controller which ensures safety and is feasible:
\begin{align*} 
    \hat{\boldsymbol{k}}_{\rm safe}(\xaux) = \underset{\hat{\boldsymbol{u}} \in \widehat{\mathcal{U}}}{\text{argmin}} \mkern9mu &
    \| \hat{\boldsymbol{u}} - \hat{\boldsymbol{k}}_{\rm p}(\xaux) \|_{\boldsymbol{W}} ^{2} \quad 
    \tag{{aGB}-QP} \label{eq:adaptive-qp}
    \\
    \text{s.t.  } 
    & \eqref{eq:h_hat_invariance}, \ \eqref{eq:hb_hat_invariance},
\end{align*}
with matrix ${\boldsymbol{W} \succ 0}$ weighing the inputs related to $\x$ and $\boldsymbol{\theta}$.
The augmented primary controller is
\begin{align}
    \hat{\boldsymbol{k}}_{\rm p}(\xaux) \triangleq
    \begin{bmatrix}
        \boldsymbol{k}_{\rm p}(\x)^\top \nspace{10}
        \boldsymbol{k}_{\rm \theta}(\x,\boldsymbol{\theta})^\top
    \end{bmatrix}^\top.
\end{align}
While $\boldsymbol{k}_{\rm p}$ is prescribed, the desired auxiliary control given by $\boldsymbol{k}_{\rm \theta}$ is in fact a design variable.
Therefore, to encourage set expansion, we propose an approach to adapt $\Btheta$ that is based on 
the gradient of 
the cost function in the optimal control problem \eqref{eq:OCP}.
We select the control direction $\boldsymbol{k}_{{\theta}}$ for the parameter $\boldsymbol{\theta}$ which increases $\hat{h}_{\rm b}(\phiauxT)$ the most, pushing the terminal point of the flow deeper inside $\Cbaux$.
This allows $\xaux$ to evolve farther from $\Cbaux$, improving the expansion.
\begin{corollary}[Adaptation for Set Expansion] \label{cor:adaptation_law}
The controller 
\begin{align}
\boldsymbol{k}_{\rm \theta}(\x,\boldsymbol{\theta}) 
= \gamma \cdot 
\boldsymbol{\Pi} \big ( \nabla \hat{h}_{\rm b}(\phiauxT)\widehat{\boldsymbol{\Phi}}(T,\xaux)\widehat{\boldsymbol{g}}(\xaux) \big)^\top, \label{eq:greedy_theta_dir}
\end{align}
adapts the parameters $\boldsymbol{\theta}$ in the direction of steepest ascent of $\hat{h}_{\rm b}(\phiauxT)$, for ${{\boldsymbol{\Pi}} \triangleq[\boldsymbol{0}_{p \times m}\nspace{6} \boldsymbol{I}_{p}]}$ and ${\gamma > 0}$, and thus locally increases the cost function of \eqref{eq:OCP}.
\end{corollary}
\section{Numerical Simulations}
\begin{figure}[t]
\hspace{-.27cm}
\centering
\begin{minipage}{0.175\textwidth}
\centering
\includegraphics[width=\linewidth]{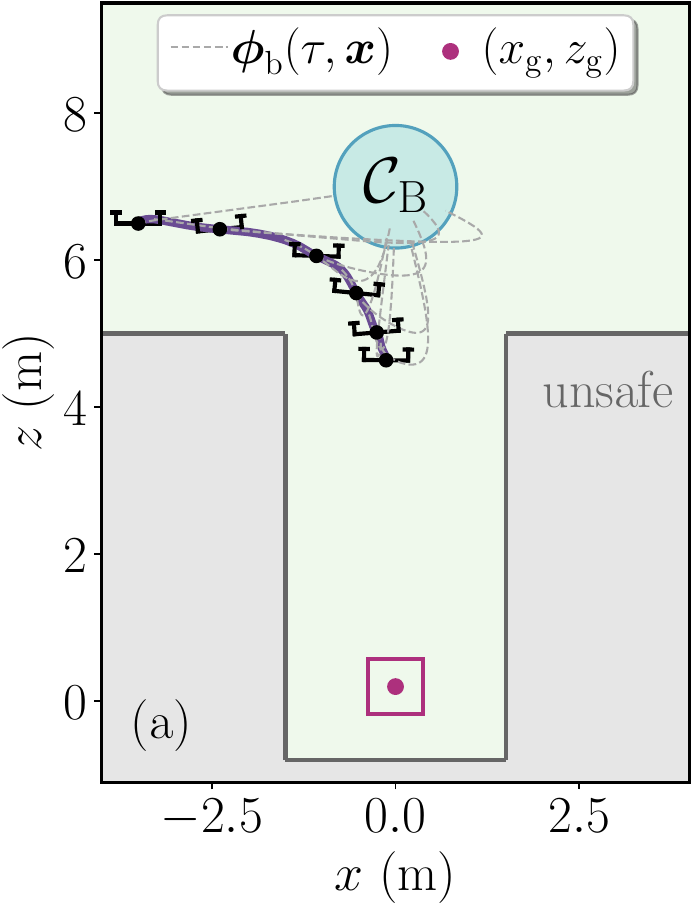}
\end{minipage}
\begin{minipage}{0.15\textwidth}
\centering
\includegraphics[width=\linewidth]{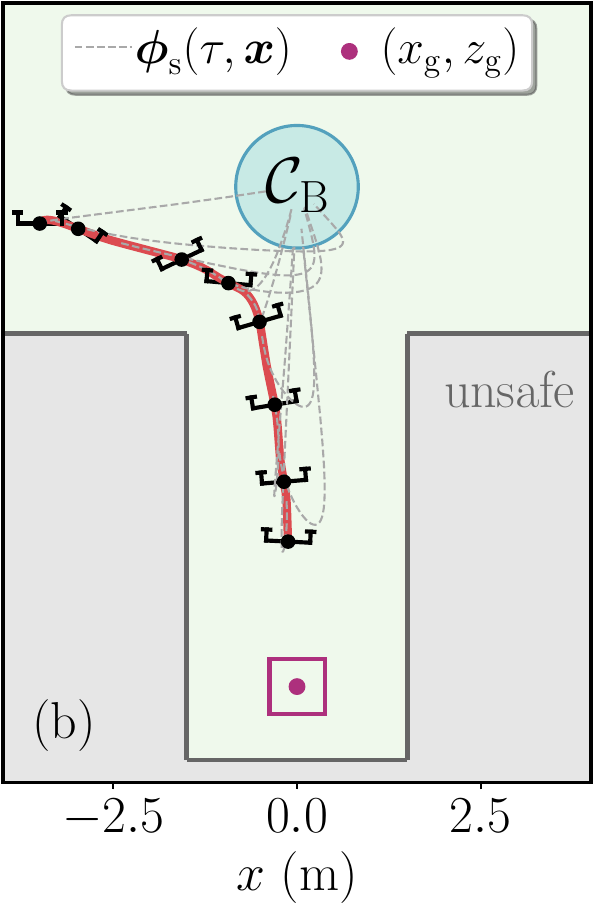}
\end{minipage}
\begin{minipage}{0.15\textwidth}
\centering
\includegraphics[width=\linewidth]{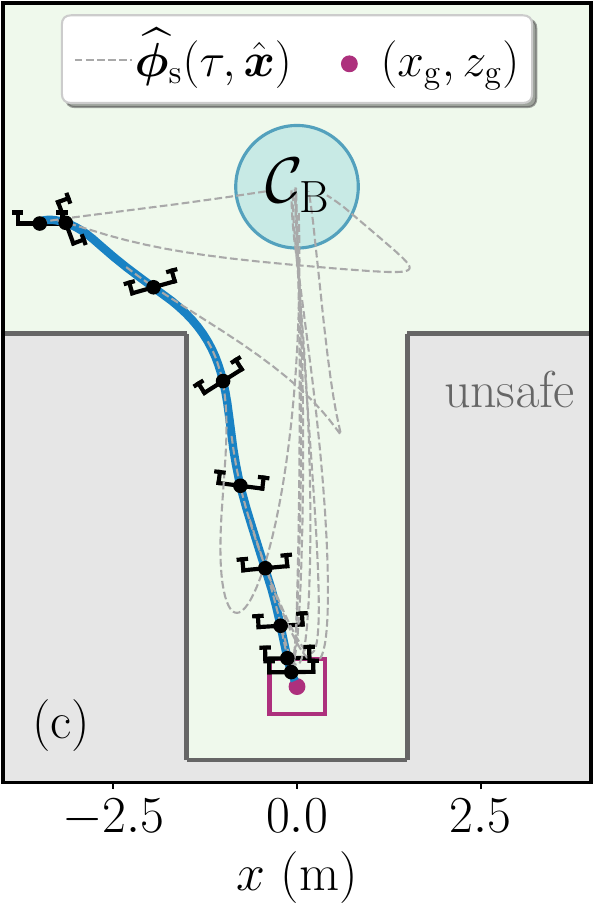}
\end{minipage}
\includegraphics[width=\linewidth]{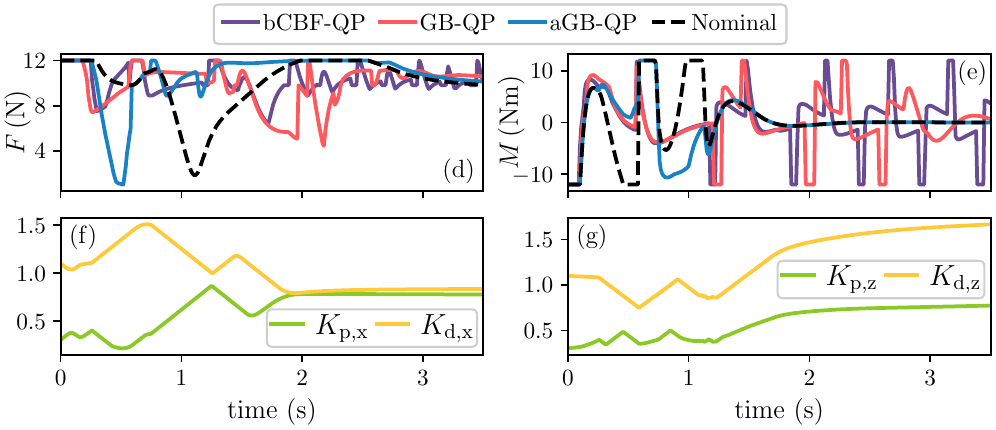}
\caption{Simulation results of the safe quadrotor landing problem, comparing the standard bCBF (purple), the generalized bCBF (red), and the adaptive bCBF (blue).
All three approaches maintain safety (\textbf{a,$\nspace{2}$b,$\nspace{2}$c}) and obey the input constraints (\textbf{d},$\nspace{2}$\textbf{e}), yet differ in task performance. The bCBF is unable to reach the goal due to the terminal constraint, as the flows (dashed gray) must reach $\Cb$ (\textbf{a}). The \eqref{eq:gen-qp} also does not reach the goal, but achieves a much larger expansion of $\Cb$ than the \eqref{eq:bcbf-qp} (\textbf{b}). Using the adaptation law in \Cref{cor:adaptation_law}, the adaptive approach reaches the goal via superior expansion of $\Cb$ (\textbf{c}) by modifying the backup controller gains (\textbf{f,$\nspace{2}$g}).
}\label{fig:quadrotor_main}
\vskip - 4mm
\end{figure}

\casestudy[Quadrotor]\label{ex:quadrotor}
Consider a planar quadrotor
\begin{align}
\begin{aligned}
    \ddot{x} \!=\! F \sin(\theta)/m, \sspace
    \ddot{z} \!=\! F \cos(\theta)/m \!-\! g_{\rm D},\sspace
    \ddot{\theta} \!=\! -M/J,
    \label{eq:quadrotor}
\end{aligned}
\end{align}
with state vector ${\x = [x \nspace{5}z\nspace{5} \theta\nspace{5} \dot{x}\nspace{5} \dot{z}\nspace{5} \dot{\theta}]^\top \!\in\! \mathcal{X}\! \subset \!\mathbb{R}^2\! \times \!\mathbb{S}^1\! \times \!\mathbb{R}^3}$.
The states $x$ and $z$ denote the horizontal and vertical positions,
and $\theta$ is the pitch angle.
The bounded inputs are the thrust
${F \in [0,F_{\rm max}]}$ and moment ${M \in [-M_{\rm max}, M_{\rm max}]}$. Here, ${g_{\rm D}}$ is gravitational acceleration, ${m}$ is the mass of the quadrotor, and ${J}$ is the principal moment of inertia about the $y$-axis.

We consider a safe landing scenario in which the quadrotor must reach a goal location ${(x_{\rm g},z_{\rm g})}$ whilst avoiding the walls and floor of a narrow landing zone (see \Cref{fig:quadrotor_main}).
The safety constraint defines the unsafe walls and floor as a polytopic region ${\Cs \!\triangleq\! \big\{ \x \!\in\! \mathcal{X} \!:\! h(\x) \!=\! \max\{\min\{h_1,h_2,h_3\},h_4 \} \!\geq\! 0 \big\}}$,
where ${h_1(\x) = x - x_{\rm L}}$, ${h_2(\x) = x_{\rm R} - x}$, ${h_3(\x) = z - z_{\rm L}}$, ${h_4(\x) = z - z_{\rm T}}$ with ${x_{\rm R} > x_{\rm L}}$ and ${z_{\rm T} > z_{\rm L}}$. The safe set is then approximated using a smooth approximation of the minimum and maximum functions, as done in \cite{molnar2025polytope}. The nominal controller is a feedback controller that attempts to reach the goal position but ignores the landing obstructions.

The backup controller is a proportional-derivative feedback controller that drives the system \eqref{eq:quadrotor} to the equilibrium point ${\x^* \!=\! [x_{\rm b}\nspace{5}z_{\rm b}\nspace{5}0\nspace{5}0\nspace{5}0\nspace{5}0]^\top \!\!\in\! \Cs}$ such that ${\boldsymbol{k}_{\rm b}(\x) \!=\! [F_{\rm b}, M_{\rm b}]}$:
\begin{align}
    a_{x} &= -K_{{\rm p},x}(x - x_{\rm b}) -K_{{\rm d},x}\dot{x},\\
    a_{z} &= -K_{{\rm p},z}(z - z_{\rm b}) -K_{{\rm d},z}\dot{z} + g_{\rm D}, \\
    \theta_{\rm d} &= {\rm sat}\big(\arctan({a_x}/{a_z})\big), \\
    F_{\rm b} &= {\rm sat}\big(m\sqrt{a_x^2 + a_z^2}\big), \\
    M_{\rm b} &= {\rm sat}\big(J (K_{{\rm p},\theta}(\theta - \theta_{\rm d}) + K_{{\rm d},\theta}\dot{\theta}) \big).
\end{align}
Here, ${K_{{\rm p},x},K_{{\rm d},x},K_{{\rm p},z},K_{{\rm d},z},K_{{\rm p},\theta},K_{{\rm d},\theta} >0}$ are gains. The backup set is a sublevel set of a Lyapunov function centered on $\x^*$, obtained by linearizing \eqref{eq:quadrotor}
under ${\ub}$ which yields
\begin{align*}
    \Cb \triangleq \big\{ \x \in \mathcal{X} \!:\! h_{\rm b}(\x) = \rho -  (\x - \x^*)^\top \boldsymbol{P} (\x -\x^*)\big\},
\end{align*}
for a positive definite matrix
${\boldsymbol{P}}$ and ${\rho \!>\! 0}$. For $\Cb$ and $\ub$ to be a valid backup set and controller pair (Def.~\ref{def:backup}), the gains and $\rho$ must be carefully chosen to ensure that 
$\ub$ does not saturate within $\Cb$, that $\ub$ renders $\Cb$ forward invariant, and that ${\Cb\!\subset\!\Cs}$. Satisfaction of these conditions (the details of which are omitted for space constraints\footnote{\label{footnote:gains}Note that sufficient conditions on the gains and $\rho$ can be obtained by using 
\cite[Sec. 3.7.3]{boyd1994linear} in conjunction with \cite[Thm. 4.7]{khalil2002nonlinear}.}) results in a restrictive set of possible gains and level set sizes. 

\Cref{fig:quadrotor_main} compares\footnote{The simulation uses ${g_{\rm D} \!=\! 9.81 \nspace{2}{\rm m/s^2}}$, ${J \!=\! {\rm 0.25 \nspace{2} kg \nspace{1} m^2}}$, ${m \!=\! {\rm 1 \nspace{2} kg}}$, ${F_{\rm max} \!=\! 12 \nspace{2} {\rm N}}$, ${M_{\rm max} \!=\! 12 \nspace{2} {\rm Nm}}$, ${\rho \!=\! 0.22}$, ${\varepsilon \!=\! 0.01}$, ${T \!=\! 6.75 \nspace{2} {\rm s}}$, ${\gamma \!=\! 10}$.} the standard backup set method (\Cref{sec:bCBF}) with the proposed generalized and adaptive approaches (\Cref{sec:generalbCBF}, \Cref{sec:adaptive}). While all three approaches guarantee safety of \eqref{eq:quadrotor} in the presence of strict input bounds, the three controllers achieve differing levels of set expansion. The \eqref{eq:bcbf-qp} is unable to substantially expand $\Cb$ due to the restrictions imposed on the gains. By separating the expanding controller and the backup controller, the \eqref{eq:gen-qp} is able to move much farther away from $\Cb$, and thus gets much closer to the goal. The \eqref{eq:adaptive-qp} uses parameters ${\boldsymbol{\theta} = [K_{{\rm p},x}\nspace{5}K_{{\rm d},x}\nspace{5}K_{{\rm p},z}\nspace{5}K_{{\rm d},z}\nspace{5}K_{{\rm p},\theta}\nspace{5}K_{{\rm d},\theta}]^\top}$ and adapts them according to \Cref{cor:adaptation_law}. This method enables the quadrotor to safely reach the goal via even better expansion of $\Cb$. Notice in panels (f) and (g) that when the quadrotor achieves a horizontal position near $x_{\rm b}$, the gains governing $x$ remain constant, while the gains governing $z$ increase to allow the quadrotor to move farther down towards the goal.
\section{Conclusion}\label{sec:conc}
We presented a novel approach for guaranteeing the safety of input-bounded nonlinear systems by enhancing the backup set method. We proved that separating the expanding controller from the backup controller preserves the controlled invariance of the expanded set, resulting in substantially larger safe sets. Further, our framework allows for online adaptation of the expanding controller, which was demonstrated to be highly beneficial in a complex quadrotor landing problem.
\end{spacing}

\vskip - 4mm
\begin{spacing}{0.92}
\bibliographystyle{ieeetr}
\bibliography{refs}
\end{spacing}

\end{document}